\documentclass[copyright,creativecommons]{eptcs}
\providecommand{\event}{GandALF 2011} % Name of the event you are submitting to
\usepackage{breakurl}             % Not needed if you use pdflatex only.

\usepackage[latin1]{inputenc}
\usepackage{graphicx}
\usepackage{listings}
\usepackage{array}
\usepackage{multicol}
\usepackage{multirow}
\usepackage{amsmath}
\usepackage{amssymb}
\usepackage{amsthm}
\usepackage{graphicx}
\usepackage{stmaryrd}
\usepackage{caption}
\newcolumntype{V}[1]{>{\small \raggedright}p{#1}}
\newcolumntype{N}{>{\small}l}
\newtheorem{definition}{Definition}
\newtheorem{theorem}{Theorem}

\title{Improving BDD Based Symbolic Model Checking with Isomorphism Exploiting Transition Relations}
\author{Christian Appold
\institute{Chair of Computer Science V\\
University of Würzburg\\
Würzburg, Germany}
\email{appold@informatik.uni-wuerzburg.de}
}
\def\titlerunning{Symbolic Model Checking with Isomorphism Exploiting Transition Relations}
\def\authorrunning{Christian Appold}
\begin{document}
\maketitle

\begin{abstract}
Symbolic model checking by using BDDs has greatly improved the applicability of model checking. Nevertheless, BDD based symbolic model checking can still be very memory and time consuming. One main reason is the complex transition relation of systems. Sometimes, it is even not possible to generate the transition relation, due to its exhaustive memory requirements. To diminish this problem, the use of partitioned transition relations has been proposed. However, there are still systems which can not be verified at all. Furthermore, if the granularity of the partitions is too fine, the time required for verification may increase. In this paper we target the symbolic verification of asynchronous concurrent systems.  For such systems we present an approach which uses similarities in the transition relation to get further memory reductions and runtime improvements. By applying our approach, even the verification of systems with an previously intractable transition relation becomes feasible.
\end{abstract}

\section{Introduction}
The presence of concurrent software is steadily increasing due to the shift towards multi-core CPUs. This software consists of several parallel threads, which are executed asynchronously and interleaved. Some models for inter-thread communication exist, but the most flexible and prominent one is the use of fully shared variables. Well-known programming APIs like the POSIX pthread model or the WIN32 API support this model of communication. Unfortunately, concurrent software often is very error-prone, and bugs tend to be subtle and are hard to detect. Thus, to enable its use in safety-critical areas, reliable techniques to verify the correct operation of concurrent software are mandatory. One formal verification technique which has been proven to be successful in the verification of concurrent systems is temporal logic model checking  \cite{TempModCheckEmCl}, \cite{TempModCheckQuSif}. There, desired properties of a system are formulated in a temporal logic (like CTL \cite{BranchingTimeLogic} or LTL \cite{LTLPnueli}), and the state-space of the system is investigated exhaustively to validate these properties. A very effective model checking technique is symbolic model checking \cite{ClarkePeledGrumModCheck00}, \cite{SymbModCheckMcM93} based on Binary Decision Diagrams (BDDs) \cite{GraphBasedAlgorithmsBryant86}.

Nevertheless, BDD-based model checking is often still very memory and time consuming. This sometimes circumvents the successful verification of systems. The main reason for the large memory requirements of symbolic model checking is often the huge size of the BDD representing the transition relation. Therefore, some methods have been proposed to diminish this problem. Originally a monolithic transition relation consisting of a single BDD was used. Due to the large size of this BDD, the authors of \cite{SymbModCheckWiPartTransRel91} suggested to use partitioned transition relations.  There, the transition relation is split into several pieces and each of these pieces can often be represented by a small BDD. Pieces of partitioned transition relations of asynchronous systems frequently  possess many identity patterns for identity transformations of state variables. In \cite{EffSollMatrixDiagMill01} and \cite{SatuBasedReachConDis05} the removal of such identity patterns has been suggested to reduce the memory overhead. In this paper we target the symbolic verification of asynchronous concurrent systems, like e.g. concurrent software. We present a new memory saving approach to store the   transition relation with BDDs. It allows to exploit similarities in the BDDs of the component transition relations. Additionally, identity patterns are removed, too. Furthermore, we introduce an algorithm that enables the efficient use of our new technique for model checking.  Our experimental results show (see section \ref{sec:ExpRes}) that this can lead to significant memory and runtime improvements. The approach is not restricted to asynchronous systems, but can be used for synchronous systems as well. To our knowledge, this is the first paper where similarities in the transition relation of components of a system are exploited that way.

The rest of this paper is organized as follows. In the next section we present some background information. We introduce our model of an asynchronous concurrent system (\ref{subsec:modelchecking}) and give a short introduction into BDDs (\ref{subsec:BDD}), symbolic state-space generation (\ref{subsec:symbStateSpace}), and symbolic representations of transition relations and related work (\ref{subsec:transRelation}). Thereafter, Section \ref{sec:Parametric} presents our new approach to store the transition relation and in Section \ref{sec:Exploration} we exemplify an efficient algorithm to build the AND of an ordinary BDD and our new data structure. Experimental results which demonstrate the efficiency of our new approach can be found in Section \ref{sec:ExpRes}. The paper closes with a conclusion and an outlook to future work.

\section{Background}
\subsection{Asynchronous Concurrent Systems}
\label{subsec:modelchecking}
In this paper we target finite state asynchronous concurrent systems $M^m=(S,R,S_0)$, where $S$ is the finite set of possible states, $S_0 \subseteq S$ is the set of initial states and $R$ is the transition relation. We assume that an asynchronous system $M^m$ is composed of $m>1$ components, and a state $s \in S$ is a  tuple $s=(\vec g, l_{1},..., l_{m})$. Thus, a system state consists of the values $\vec g$ of all global shared variables (not associated with any component) and the $local \, state \, l_{i}$ of each component $i \in \{1,...,m\}$ (i.e. values of all local variables of component $i$). The transition relation is defined as $R=\{(x,x')|$ $x \in S \land  x' \in S \land$ state $ x'$ can be reached from state $x$  in a single step$\}$.

The execution model of a system $M^m$ is that of interleaved asynchrony. Only one component can execute a transition at a time and a transition of a component $i$
only depends on and only changes the values of the shared variables $\vec g$  as well as its own local state $l_i$.  That means, a component has neither read nor write access to local variables of other components. We denote this frequently occurring behavior as \textit{transition locality}. Let $R_{i_{P}}$ be a relation with $R_{i_{P}}=\{((\vec g, l_{i}),(\vec g', l_{i}'))|$  $(\vec g', l_{i}')$ results from $(\vec g, l_{i})$ by executing a single step of \textit{i}$\}$ and let $R_i$ be  the transition relation of component $i$ that contains the transitions executable by component $i$. In systems with transition locality the following holds $\forall i \in \{1,...,m\} : R_{i}=\{((\vec g, l_{1},..., l_{m}),(\vec g', l_{1}',..., l_{m}'))|\forall j \neq i: \: l_{j}'=l_{j} \land  ((\vec g, l_{i}),(\vec g', l_{i}')) \in R_{i_{P}}\}$ and $R=\bigcup_{i \in m}R_i$. An example for this system type is the tremendous importance gaining concurrent software for multi-core architectures with threads which communicate via shared variables. Also the subtype of concurrent software with replicated threads is most relevant in practice. A formal definition of this system type can be found in \cite{KaiserKroeningWahl10DynCuttoff}.

\subsection{Binary Decision Diagram (BDD)}
\label{subsec:BDD}
Decision diagrams are used in symbolic model checking to store sets of states as well as the transition relation of a system.  A \textit{binary decision diagram} (BDD) \cite{GraphBasedAlgorithmsBryant86} for \textit{N}-variables can be used to encode a function $f:\{0,1\}^{N} \mapsto \{0,1\}$.
\vspace*{0.5cm}
\begin{definition}
A BDD is an acyclic directed graph with a single root vertex and two types of vertices, $nonterminal$ vertices and $terminal$ vertices. Each $nonterminal$ vertex $v$ is labeled by a variable $var(v)$ and has two successors $low(v)$ and $high(v)$. A $terminal$ vertex $v$ is labeled by a value $value(v) \in \{0,1\}$.
\end{definition}
\vspace*{0.5cm}
As we did in this paper, most often \textit{reduced ordered binary decision diagrams} ROBDDs \cite{GraphBasedAlgorithmsBryant86} are used. ROBDDs are a canonical representation for boolean functions. Canonicity is achieved by using two restrictions for BDDs. There should be no isomorphic subtrees or redundant vertices in the diagram, and  the variables should appear in the same order along each path from the root vertex to a $terminal$ vertex. The same order for the variables along each path is ensured by using a total ordering $\prec$ on the variables that label the vertices in a BDD. Then $var(u) \prec var(v)$ is required for any vertex $u$ in the diagram that has a $nonterminal$ successor $v$. One can decide whether a particular truth assignment to its variables makes a function represented as a BDD true, or not, by traversing the graph from the root vertex to a terminal vertex. The value of a reached $terminal$ vertex is the value of the function for the given variable assignment.

\subsection{Symbolic State-Space Generation}
\label{subsec:symbStateSpace}
As mentioned in the last section, BDDs are used in symbolic model checking to store sets of states as well as the transition relation of a system. A set of states $Z$ can be encoded with a BDD through its characteristic function $\chi_{Z}$. If the shared states $\vec g$ of an asynchronous system with $m$ components can be encoded with $n_g$ boolean variables and the local states of a component $i$ with $n_{l_{i}}$ boolean variables, then a BDD for $N=n_g + \sum_{i=1}^m n_{l_{i}}$ variables can be used to store sets of system states. To encode the transition relation with a BDD, transitions between states, instead of single states, have to be encoded. Therefore, a BDD for twice as many variables as for BDDs that encode sets of states is necessary and the transition relation can be encoded with a BDD for \textit{2N}-variables. There \textit{N}-variables are needed for the \textit{from}-state and also \textit{N}-variables for the \textit{target}-state of a transition. As BDD variable ordering for the \textit{2N}-variables, all possible permutations are applicable. But it is widely acknowledged that variable ordering with interleaving of the corresponding \textit{from}- and \textit{target}-state variables is often the most efficient variable ordering by terms of nodes required to store the transition relation. Thus, we consider only interleaved variable ordering in this work. In interleaved variable ordering the corresponding \textit{from}- and \textit{target}-state variables are next to each other in a BDD.

This paper targets on forward reachability analysis. There, the image computations are forward images and the forward image for a set of states $Z$ is defined as: $Image(Z)=\{x' | \exists x \in Z, (x,x') \in R\}$. In forward reachability analysis state-space search starts with the set of initial states $S_0$. The set of reachable states is the minimal set satisfying $Z \supseteq S_0$ and $Z \supseteq Image(Z)$ which can be computed through iterated forward image calculations. The traditional approach for symbolic state-space generation, which we also used within this paper, uses breadth-first iterations. Each breadth-first iteration consists of an image computation with the entire transition relation $R$ of a system.  At the $i$th iteration all states with distance less or equal $i$ from the initial states have been explored.

\subsection{Symbolic Representations of Transition Relations and Related Work}
\label{subsec:transRelation}
A monolithic transition relation of a single BDD is often intractably large. Therefore, the use of partitioned transition relations has been proposed in \cite{SymbModCheckWiPartTransRel91}.  Partitioned transition relations consist of conjunctions or disjunctions of a number of pieces of the single BDD. These pieces can often be represented by a small BDD. In this paper we consider asynchronous concurrent systems and use disjunctive partitioned transition relations. A component-wise disjunctively partitioned transition relation for an asynchronous system with $m$ components is composed of the transition relations $R_i$ of the components, and can be written as $R=R_{1} \vee R_{2} \vee ... \vee R_{m}$. In this work we consider only systems with transition locality (see section \ref{subsec:modelchecking}). Our method further reduces the memory requirements of the partitioned transition relation approach through exploiting similarities in the transition relation of the components. For the use of partitioned transition relations, it's worth mentioning that a too fine granulated transition relation may not be the best choice. As long as the BDDs don't become too large, it is better to combine several transitions in one disjunct. In this way, fewer BDD nodes may be needed and also image calculation can possibly be accelerated. In \cite{Ranjan95EfficientBDDalgorithms} the authors presented and investigated an approach where the partitions of partitioned transition relations can consist of several transitions. Their experimental results confirm that larger partitions lead to big runtime savings. But they also observed an increase in the number of BDD nodes for coarser partitioned transition relations. By considering similarities in the transition relations of the components our approach allows to build much coarser partitions of transitions. Additionally, in the presence of  large isomorphic subgraphs no strong increase in the total number of BDD nodes occurs. Thus, our approach can reduce the runtime without causing an increase of the memory requirements.

Transition relations of asynchronous systems often contain many identity patterns. As introduced in \ref{subsec:modelchecking}, if a component $i$ executes a transition in a system with transition locality, then the local states for all other components $j \neq i$ remain unchanged. Therefore, the BDD for the transition relation $R_i$ of component $i$ contains identity patterns for the local state bits of all other components $j \neq i$.  An example of an identity pattern can be found in Figure \ref{fig:identity}. There level $k$ contains a vertex of a \textit{from}-state and level $k+1$ a vertex of the corresponding \textit{target}-state. According to Figure \ref{fig:identity}, if the vertices at level $k$ and $k+1$ get assigned different values, then the BDD evaluates to $0$. That means, if a BDD for a transition contains an identity pattern for a variable, the variable doesn't change its value when the transition is executed. To avoid the memory overhead to store identity patterns, \cite{SatGenClassMod04} introduces an approach which uses reduced matrix diagrams (MxDs) \cite{EffSollMatrixDiagMill01} without identity nodes for the transition relation. The authors of \cite{SatuBasedReachConDis05} suggested to use a new identity reduction rule for MDDs \cite{MultiValuedDecDiag98} to get fully identity reduced MDDs for the transition relation. These papers just present approaches for identity reduction, but no method to use similarities in the transition relations of components. A technique to exploit sharing in BDDs for regular circuits that differ only in their support variables has been presented in \cite{Bryant03SymbRepresOrdFuncTempl}. Similar to our approach a remapping of input variables is used there. But such a remapping can not be used for BDDs of transition relations of components in asynchronous concurrent systems. The reason is different positions of identity patterns in the BDD variable ordering for different components. Additionally, they always expand a BDD with modified input variables before performing a BDD operation. This is very time consuming and can even be intractable for large transition relations. To solve this problem, we present in section \ref{sec:Exploration} an efficient algorithm for boolean operation calculation with our new BDD type, which avoids the expansion to a normal BDD.

\section{Transition Locality Exploiting BDDs (TLEBDDs)}
\label{sec:Parametric}
In this section we present our new approach to store the transition relation of systems with transition locality (see section \ref{subsec:modelchecking}). It makes use of the circumstance that BDDs for subsets of the transition relation may have a very similar structure, if the transition relation is split component-wise in partitioned transition relations. To exploit those similarities and to reduce the memory requirements of transition relations we suggest to use \textit{Transition Locality Exploiting BDDs (TLEBDDs)}. A TLEBDD consists of a normal BDD (see section \ref{subsec:BDD}) and a mapping list. For a system with $m$ components, the transition relation of a component can be represented by a BDD with $2 \cdot N$ variables, where $N=n_g + \sum_{i=1}^m n_{l_{i}}$. In the rest of the paper we assume that the BDD of a TLEBDD for a component $i$ is defined over the variables $X=\{x_1,x_2,...,x_{2 \cdot (n_{g} + n_{l_{i}})}\}$  and the mapping list is defined over the variables $Y=\{y_1,y_2,...,y_{2 \cdot N}\}$. We will denote the variables in $X$ as reduced variables and the variables in $Y$ as actual variables. The mapping list is necessary to map the reduced variables to the actual variables of the corresponding characteristic function $\chi_{R_{i}}$ of $R_{i}$  for which the TLEBDD has been built. For a component $i$ this mapping can be described with a function $\pi:\: \{1,2,..,2 \cdot (n_{g} + n_{l_{i}})\} \rightarrow \{1,2,..,2 \cdot N\}$ that maps mapping list entries to variable indices from $Y$.
\vspace*{0.5cm}
\begin{definition}
A $n$-\textit{mapping list} is a list over $Y$ with $n$ elements, that is \\ $[y_{\pi(1)},...,y_{\pi(n)}]$.
\end{definition}
\vspace*{0.5cm}
According to section \ref{subsec:modelchecking} the transition relation $R_i$ of a component $i$ in a system with transition locality is defined as $R_{i}=\{((\vec g, l_{1},..., l_{m}),(\vec g', l_{1}',..., l_{m}'))|\forall j \neq i: \: l_{j}'=l_{j} \land  ((\vec g, l_{i}),(\vec g', l_{i}')) \in R_{i_{P}}\}$ (see section \ref{subsec:modelchecking} for the definition of $R_{i_{P}}$) and the values of $\vec g'$ and $l_{i}'$ depend only on $\vec g$ and $l_{i}$. TLEBDDs exploit the circumstance that for every transition of a component $i$ holds $\forall j \neq i: l_{j}'=l_{j}$, and no vertices are used in the transition relation of a component $i$ for the local states of an other component $j \neq i$.
\vspace*{0.5cm}
\begin{definition}
 A TLEBDD for the transition relation of a component $i$ in a system with transition locality is a tuple $(G,b)$, where $G$ is a normal BDD and $b$ is a mapping list.
$G$ is a BDD with the $2 \cdot (n_{g}+n_{l_{i}})$ reduced variables $X=\{x_1,x_2,...,x_{2 \cdot (n_{g} + n_{l_{i}})}\}$. They are used for the bits of the shared states ($2 \cdot n_g$ bits) of the system and the local state bits ($2 \cdot n_{l_{i}}$ bits) of the component for which the TLEBDD has been built. For actual variables of the other $n-1$ components a TLEBDD implicitly assumes identity patterns.  The mapping list $b$ contains $n=\{1,2,...,2 \cdot (n_{g} + n_{l_{i}})\}$ elements and is used to map the reduced variables of $G$ to the actual variables. It contains for each position $q \in \{1,2,...,2 \cdot (n_{g} + n_{l_{i}})\}$ in the variable ordering of the BDD $G$ the associated actual variable $y_{\pi(q)}$. Thereby it holds for $q_1,q_2 \in \{1,2,...,2 \cdot (n_{g} + n_{l_{i}})\}$ with $q_1 \neq q_2$ that $y_{\pi(q_1)} \neq y_{\pi(q_2)}$.
\end{definition}
\vspace*{0.5cm}

\begin{figure}[t]
\begin{center}
\begin{minipage}[t]{5.5cm}
\begin{center}
\includegraphics[width=2.7cm]{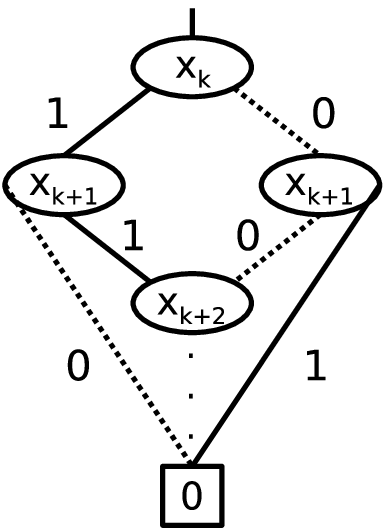}
\end{center}
\caption{Example of an identity pattern}
\label{fig:identity}
\end{minipage}
\begin{minipage}[t]{10cm}
\begin{center}
\includegraphics[width=8.5cm]{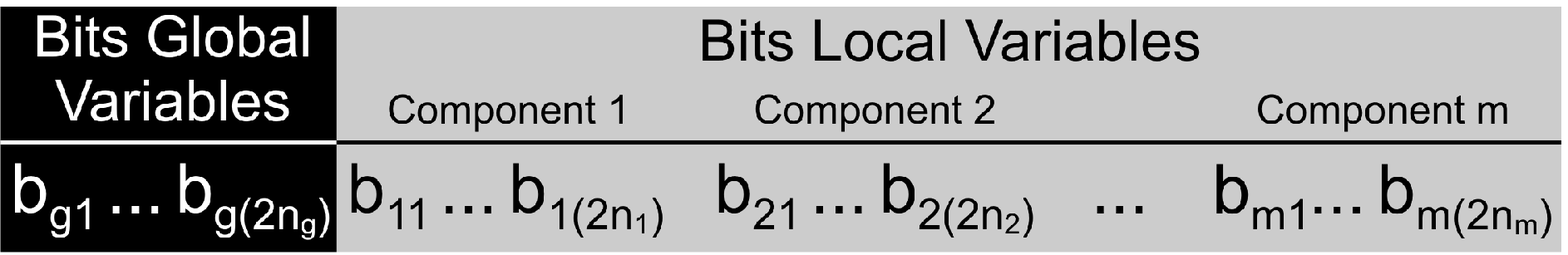}
\end{center}
\caption{Variable Ordering Concatenated}
\label{fig:indices}
\end{minipage}
\end{center}
\end{figure}

TLEBDDs can be used for the efficient representation of component transition relations. A corresponding BDD can be obtained from a TLEBDD $(G,b)$ through substitution of the reduced variables of the TLEBDD with the corresponding actual variables and the insertion of identity patterns. That means,  the TLEBDD $(G,[y_{\pi(1)},...,y_{\pi(n)}])$ and the BDD $\Delta(G[y_{\pi(1)}/x_{1},...,y_{\pi(n)}/x_{n}])$ represent the same function. Here $G[y/x]$ is the substitution of any occurrence of $x$ in $G$ with $y$ and $\Delta$ is an operation which inserts identity patterns for associated pairs of from- and target-state actual variables for which no corresponding reduced variables exist. TLEBDDs use the same reduced variables to represent the local state bits of different components. 
In the prominent special case of asynchronous systems with only one replicated component type, even all corresponding local state bits of the $m$ components can be mapped to the same reduced variables. In this way we get isomorphic subgraphs which aren't isomorphic in BDDs of ordinary partitioned transition relations, because the position of local state bits of components or of identity patterns in the variable ordering differs. This enables us to use the common property of BDD packages like Cudd \cite{Somenzi09Cudd} to store isomorphic subgraphs only once. Our experimental results in section \ref{sec:ExpRes} confirm that this can lead to enormous memory savings. TLEBDDs can be made canonical by requiring that mapping lists are ordered with respect to some strict ordering $\prec$ on the actual variables $Y$.
\vspace*{0.5cm}
\begin{definition}
A $n$-\textit{mapping list} is ordered, if $y_{\pi(i)} \prec y_{\pi(i+1)}$, for all $1 \leq i < n$.
\end{definition}
\vspace*{0.3cm}
\begin{theorem}
If $(G,b_g)$ and $(H,b_h)$ are two TLEBDDs with mapping lists which are ordered with respect to some strict ordering $\prec$ on the actual variables $Y$, then for boolean functions $g,h$ of component transition relations with $g$ represented through $(G,b_g)$ and $h$ represented through $(H,b_h)$, $g=h$ holds, if and only if $G=H$ and $b_g=b_h$.
\end{theorem}
\begin{proof}
Let $b_g=b_h=[y_{\pi(1)},...,y_{\pi(n)}]$. By expansion of the TLEBDDs we get $g_{exp}=$
\newline 
$\Delta(G[y_{\pi(1)}/x_{1},...,y_{\pi(n)}/x_{n}])$ and $h_{exp}=\Delta(H[y_{\pi(1)}/x_{1},...,y_{\pi(n)}/x_{n}])$, where $\Delta$ is defined as introduced before. Because $G=H$, we get $g_{exp}=h_{exp}$ and therefore holds $g=h$.
\newline
Be now $g=h$. Because the mapping lists have to be strictly ordered and the same actual variables have to be mapped to reduced variables, there is only one unique ordered mapping list. Thus $b_g=b_h$ holds. If $G \neq H$ would hold, then the TLEBDDs $(G,b_g)$, and $(H,b_h)$ respectively, have to have different mapping lists $b_g$ and $b_h$ that $g=h$ can be valid. Therefore also $G=H$ holds.
\end{proof}
\vspace*{0.5cm}
A TLEBDD can be built for a component $i$ through encoding of the relation $R_{i_{P}}$ by using the reduced variables instead of the actual variables. Additionally the mapping of the $2\cdot(n_g+n_{l_{i}})$ reduced variables to the $2N=2 \cdot (n_g + \sum_{i=1}^m n_{l_{i}})$ actual variables has to be stored in the mapping list. To evaluate the truth value of a particular assignment of values to the variables of a TLEBDD, its BDD has to be traversed from the root vertex to a terminal vertex similar to a BDD.
Additionally, during its traversal the information which has been stored in the mapping list has to be considered to map the reduced variables to their corresponding actual variables and to take into account the missing identity patterns.

To use TLEBDDs and ordinary BDDs for model checking, it's necessary that they can be combined through boolean operations. An approach that allows the use of the traditional BDD algorithms to combine a TLEBDD and a BDD is to adapt the TLEBDD variable ordering to the variable ordering of the BDD and to insert simultaneously the omitted identity patterns. Though this works, here the uncompressed BDD has to be built for a TLEBDD. This would cause an additional runtime overhead, which can sometimes be very large. Also, if this BDD is huge a lot of memory may be required. In the worst case this can lead to an abort of the subsequent forward image calculation and therewith the model checking run. Therefore, we developed an effective algorithm for the calculation of boolean operations which avoids to generate normal BDDs for TLEBDDs entirely. In this way the vertices of  a corresponding BDD for a TLEBDD are not needed at all, and we achieve the maximum possible memory reduction.

\section{Efficient Algorithm for Boolean Operation Calculation}
\label{sec:Exploration}
Here, we exemplify an efficient algorithm to compute the AND of a TLEBDD and a BDD. The AND of two BDDs is a very important step in forward image computation, because in every forward image computation the AND of the BDD with states which still have to be explored and the transition relation has to be calculated. Listing \ref{lst:AndBDDs} sketches our new algorithm which allows to build the AND of a TLEBDD and a BDD without building the corresponding normal BDD for the TLEBDD at all. Prior to the execution of the algorithm the  variable ordering of the reduced variables of the TLEBDD has to be adapted according to the variable ordering of the BDD .
\vspace*{0.5cm}
\begin{center}
\begin{minipage}[t]{12.5cm}
\begin{lstlisting}[language=C++,numbers=left,showspaces=false,showtabs=false,showstringspaces=false,captionpos=t,caption=Recursively  compute the AND of a TLEBDD and a normal BDD,label=lst:AndBDDs,basicstyle=\ttfamily\scriptsize,mathescape]
ANDRecursive(TLEBDDVertex TLEroot, BDDVertex BDDroot, int actualVarTLE){
   BDDVertex result = TERMINAL_CASE(TLEroot,BDDroot,actualVarTLE);
   if(result != NULL){
      return  result;} //terminal case found
   result = COMPUTED_TABLE_HAS_ENTRY(AND,TLEroot,BDDroot,actualVarTLE);
   if(result != NULL){
      return result;}  //result has already been calculated before

   if(BDDroot.variable $\prec$ actualVarTLE){
      $v$ = BDDroot.variable; 
      T = ANDRecursive(TLEroot,BDDroot$_v$,actualVarTLE);
      E = ANDRecursive(TLEroot,BDDroot$_{\overline{v}}$,actualVarTLE);}
   else{
      $v$ = actualVarTLE;
      $w$ = TLEroot.variable; 
      TLEroot$_w$ = getNextVertex(TLEroot,TLEroot$_w$,actualVarTLE);
      actualVarNew$_w$ = getNextVertexVar(TLEroot,TLEroot$_w$,actualVarTLE);
      TLEroot$_{\overline{w}}$ = getNextVertex(TLEroot,TLEroot$_{\overline{w}}$,actualVarTLE);
      actualVarNew$_{\overline{w}}$ = getNextVertexVar(TLEroot,TLEroot$_{\overline{w}}$,actualVarTLE);
      T=ANDRecursive(TLEroot$_w$,BDDroot$_v$,actualVarNew$_w$);
      E=ANDRecursive(TLEroot$_{\overline{w}}$,BDDroot$_{\overline{v}}$,actualVarNew$_{\overline{w}}$);}

   if(T == E) return T;
   R = FIND_OR_GENERATE_AND_ADD_UNIQUE_TABLE(v,T,E);
   INSERT_COMPUTED_TABLE((AND,TLEroot,BDDroot,actualVarTLE),R);
   return R;}
\end{lstlisting}
\end{minipage}
\end{center}
\vspace*{0.5cm}
One main difference of the algorithm in Listing \ref{lst:AndBDDs} to the usual AND algorithm is the use of a variable \textit{actualVarTLE} for the current actual variable of a TLEBDD vertex. This variable is necessary to achieve that only those TLEBDD and BDD vertices are evaluated together that would also be evaluated together if the AND would be done between two ordinary BDDs. In line 2 of the algorithm it is detected if a terminal case of the recursive computation has been reached. If a terminal vertex is reached in a normal BDD, then its value is the value of the represented function for the variable assignment that led to this terminal vertex. In our algorithm a terminal vertex of a TLEBDD is really a terminal vertex, if its value is $0$. If its value is $1$, possibly missing identity patterns have to be evaluated before the terminal vertex is valid. This problem can be solved by using the value of  \textit{actualVarTLE} to decide the validity of such terminal vertices during the detection of terminal cases. The value of \textit{actualVarTLE} also has to be considered during computed table accesses (see lines 5 and 25). This has to be done because different partial results of the AND operation can occur with the same TLEBDD and BDD vertices. By considering the value of \textit{actualVarTLE} these partial results can be differentiated.

\vspace*{0.5cm}
\begin{center}
\begin{minipage}[b]{14cm}
\begin{lstlisting}[language=C++,numbers=left,showspaces=false,showtabs=false,showstringspaces=false,captionpos=t,caption=Compute a successor vertex of  the current TLEroot in a TLEBDD,label=lst:NextNode,basicstyle=\ttfamily\scriptsize,mathescape]
getNextVertex(TLEBDDVertex TLEroot, TLEBDDVertex TLEroot$_{succ}$,int actualVarTLE){
   TLEBDDVertex TLEroot$_{new}$ = TLEroot$_{succ}$;
 
   if(isTerminalVertex(TLEroot) || 
     (actualVarTLE $\prec$ mappingList[TLEroot.variable])){
       TLEroot$_{new}$ = TLEroot;}
 
   return TLEroot$_{new}$;}
\end{lstlisting}
\end{minipage}
\end{center}
Line 9 decides which of the two decision diagrams has the top variable in the used variable ordering at a step of the recursion. Adjustments to  \textit{actualVarTLE} and \textit{TLEroot} for recursive calls of ANDRecursive have to be done only if \textit{actualVarTLE} is the current top variable. Otherwise, its value is kept because the current root of the TLEBDD corresponds to an actual variable which has to be evaluated later. In the \textit{else} path the new values of  \textit{actualVarTLE} (\textit{actualVarNew$_w$} and \textit{actualVarNew$_{\overline{w}}$}) as well as \textit{TLEroot} (\textit{TLEroot$_w$} and \textit{TLEroot$_{\overline{w}}$}) have to be determined according to the current value of \textit{actualVarTLE} and the mapping of the reduced variables of the TLEBDD vertices into the BDD variable ordering. Thereby the values of the new \textit{TLEroot}s are calculated with the function \textit{getNextVertex()} (see lines 16 and 18) and the new values of \textit{actualVarTLE} are calculated with the function \textit{getNextVertexVar()} (see lines 17 and 19). In the function \textit{getNextVertex()} (see Listing \ref{lst:NextNode})  \textit{TLEroot} has to keep its value, if it is already a terminal node, or if the value of  \textit{actualVarTLE} is before the  actual variable that corresponds to the reduced variable of TLEroot in the variable ordering. This is necessary, because of the missing identity patterns in a TLEBDD, and \textit{TLEroot} has to be evaluated later in the variable ordering.  Otherwise \textit{getNextVertex()} returns the successor \textit{TLEroot$_{succ}$} as the new root of the TLEBDD. The new value of \textit{actualVarTLE} is calculated by the function \textit{getNextVertexVar()} (see Listing \ref{lst:NextNodeVariable}). If the successor vertex \textit{TLEroot$_{succ}$} is a terminal vertex with value $0$, then the terminal vertex can be evaluated immediately and \textit{actualVarTLE} gets the value for a terminal vertex (see line 6). Otherwise, the function \textit{identityPatternBeforeSuccVertex()} detects if there is an actual variable for an identity pattern between \textit{actualVarTLE} and the corresponding actual variable of \textit{TLEroot} in the variable ordering. If there is such an actual variable, the function \textit{getNextActualIdentityPatternCurrVar()} calculates the next occurring actual variable of an identity pattern for a from-state and  \textit{actualVarNew} is set to this value.

These calculations can be done with the help of the mapping list and the parameter values of the functions \textit{identityPatternBeforeTLEroot()}, and \textit{getNextActualIdentityPatternCurrVar()} respectively. If no actual variable for an identity pattern exists in the variable ordering before the corresponding actual variable of \textit{TLEroot},  \textit{actualVarNew} can be set to a value for a terminal vertex if \textit{TLEroot} is a terminal vertex. When \textit{TLEroot} is no terminal vertex, \textit{actualVarNew} is set to the value of an actual variable for an identity pattern before \textit{TLEroot$_{succ}$} or to the actual variable that corresponds to the formal variable of \textit{TLEroot$_{succ}$}. By setting the value of \textit{actualVarNew} to the first variable of every occurring identity pattern, we achieve that the recursion definitely holds at each such variable. The impact of the missing identity patterns then can be considered at these recursion steps.
\vspace*{0.5cm}
\begin{center}
\begin{minipage}[t]{14cm}
\begin{lstlisting}[language=C++,numbers=left,showspaces=false,showtabs=false,showstringspaces=false,captionpos=t,caption=Compute a new value for \textit{actualVarTLE},label=lst:NextNodeVariable,basicstyle=\ttfamily\scriptsize,mathescape]
getNextVertexVar(TLEBDDVertex TLEroot, TLEBDDVertex TLEroot$_{succ}$,int actualVarTLE){
   int actualVarNew;
        
   if((TLEroot$_{succ}$.index==CONST_INDEX) && (TLEroot$_{succ}$.value==0)){ 
   {//a terminal vertex with value $0$ can be evaluated immediately
       actualVarNew = CONST_INDEX;}
   else{
      //decide if there is an identity pattern before TLEroot 
      //that has to be evaluated
      if(identityPatternBeforeTLEroot(TLEroot,actualVarTLE)==TRUE){
         actualVarNew = 
             getNextActualIdentityPatternCurrVar(TLEroot,actualVarTLE);}
      else{
         if(TLEroot.index==CONST_INDEX){ 
           actualVarNew = CONST_INDEX;}     
         else{
           if(identityPatternBeforeTLEroot(TLEroot$_{succ}$,actualVarTLE)==TRUE){
             actualVarNew = 
              getNextActualIdentityPatternCurrVar(TLEroot$_{succ}$,actualVarTLE);}
           else{
              actualVarNew = mappingList[TLEroot$_{succ}$.variable];}}}}}    

   return actualVarNew;}
\end{lstlisting}
\end{minipage}
\end{center}
\vspace*{0.5cm}
If a step of the recursion has finished, the calculated subgraphs $T$ and $E$ have to be combined and the result has to be returned. The return value is determined in lines 23 and 24 of Listing \ref{lst:AndBDDs}. If the top variable of the recursion step isn't a variable for an identity pattern, the return value can be calculated as it is done in the algorithm for the AND between two normal BDDs. When the top variable is a variable for an identity pattern, the recursion definitely holds at this recursion step and the variable is a from-state variable of the identity pattern. Here the impact of the missing identity patterns to the result of an AND operation is taken into account. Figure \ref{fig:resultC} illustrates the effect of identity patterns on the result calculation. In principle three different cases have to be considered. They are marked with a, b, and c, and $x_k$ is the top and also from-state variable of an identity pattern. For each case the Figure shows in the left the result of the recursion at this step if the AND had been calculated with identity patterns. On the right side the result which our algorithm returns for TLEBDDs is shown. Except for the first case (a), two subgraphs are shown as solutions for our algorithm. There are two different subgraphs because of different optimizations that we used. Generally, after forward image calculations first the from-state variables are existentially abstracted and after that the target-state variables are shifted to their corresponding from-state variables. This is done with two different functions calls. Beneath the image calculation itself, these functions often need a lot of runtime. If TLEBDDs are used the abstraction of the from-state variables can be done easily and with little runtime overhead for variables for identity patterns. To do this there have to be inserted no vertices for the from-state level but only the correct remaining subgraph without the from-state vertex has to be built. Therefore we developed a version where from-state variables for identity patterns are abstracted away immediately. The outcome of the result combination with this immediate abstraction are captioned with \textit{exist abst.} in Figure \ref{fig:resultC}. Also we observed that the shift to the from-state variables often needs a lot of runtime. We developed a second method for result combination, where the target-state variables are immediately shifted to their corresponding from-state variables. This can be done easily for identity patterns. For the verification experiments we implemented the immediate shift for all variables. For non identity variables there is more work to do to get the correct subgraphs. As our experimental results show, the immediate shift leads to very large runtime and memory improvements. Thus interleaved variable ordering is very efficient in combination with the immediate shift to the from-state variables. In the first case (a) the target-state vertices at level $k+1$ have as one successor the same subgraph $T$. This corresponds to the case  where $T$ and $E$ are equal in our algorithm (see line 23 in Listing \ref{lst:AndBDDs}). When $T$ equals $E$, the subgraph $T$ can be returned regardless if an immediate abstraction of the from-state variables or an immediate shift to the from-state variables was done. If $T$ and $E$ are not equal, the result is calculated in line 24. Here two different cases can occur  in the presence of identity patterns. The one is numbered with (b) in Figure \ref{fig:resultC} and there different subgraphs $T$ and $E$ exist for the identity paths. After abstraction of the from-state variables the subgraph with root variable $x_{k+1}$ is the correct result. If an immediate shift is done, the result is the subgraph with root $x_{k}$.
In the last case (c) only a system state exists for the value $1$ of $x_k$ for the current variable assignment (the same behavior can occur with value $0$ for $x_k$). Here our algorithm also returns the subgraph with root $x_{x+1}$ or $x_{x}$ in dependency of the chosen result combination strategy. After the result has been calculated for a recursion step, it is inserted into the computed table (see line 25) and returned.
\begin{figure}[t]
\centering
\includegraphics[width=14cm]{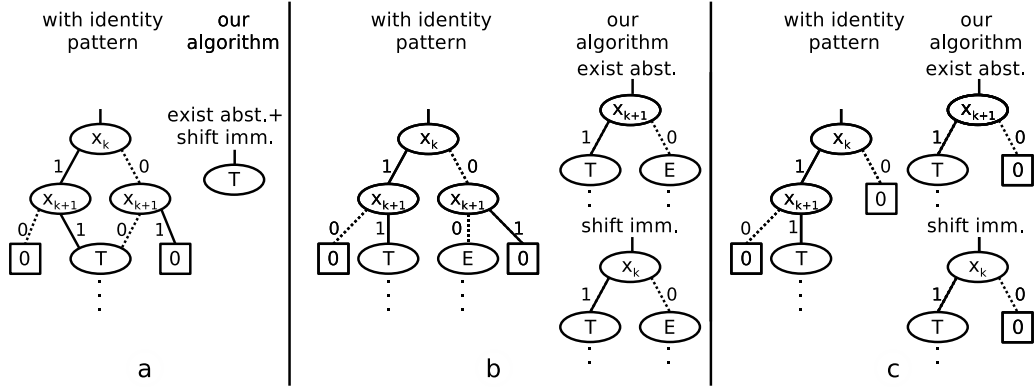}
\caption{Handling of identity patterns during combination of the subgraphs $T$ and $E$}
\label{fig:resultC}
\end{figure}

\section{Experimental Results}
\label{sec:ExpRes}
In this section we present the results of our verification experiments. The experiments run on an Intel Pentium Core 2 CPU with 2.4 GHz and 3 GB main memory by using a single core. The verification experiments have been done with an adapted version of the symbolic model checker Sviss \cite{WahlBlancEmerson08Sviss}, which uses the Cudd BDD package \cite{Somenzi09Cudd}. For our experiments we have chosen the variable ordering concatenated for the bits of the components in the BDDs, because it is efficient for asynchronous systems. Figure \ref{fig:indices} presents this variable ordering. The first bits in this variable ordering are  $b_{g1}$ to $b_{g(2n_g)}$. They denote the from-state and target-state bits for the shared variables of a system state. The bits $b_{i j}$ denote the $j$th bit of component $i$. All experiments have been done with partitioned transition relations with identical sets of transitions in every partition for the different transition relation types. All testcases describe asynchronous systems with replicated components. In the following tables the number of replicated components can be found in the column \textit{Problem} after the name of the verification benchmark. \textit{Number of BDD Nodes} is the  largest number of live BDD nodes that appeared during a verification experiment. This is the memory bottleneck of a verification experiment, because the model checker has to store this number of BDD nodes to finish verification successfully \cite{YangBryantBiere98CavBDDPerformanceStudy}. \textit{Time} is the runtime of a verification experiment, where s, m and h are abbreviations for seconds, minutes and hours. In Table \ref{table:ForwardReachability} we show experimental results for forward reachability analysis. Experimental results for a standard partitioned transition relation, a TLEBDD as transition relation (from-state variables are abstracted immediately for identity patterns here (see section \ref{sec:Exploration})) and a TLEBDD as transition relation where we immediately shift the target-state variables to its corresponding from-state variables are presented there. With the immediate shift we achieve significant runtime improvements and the memory gain can be maximized. One reason for the memory gain is that vertices which can be saved in a TLEBDD are not needed for the intermediate result BDD before the shift to the from-state variables. For the experiments in Table \ref{table:ForwardReachability} we used a timeout of $24$ hours.

\begin{table}[htbp]
\caption{Verification results for forward reachability analysis}
\centering
\begin{tabular}{|c|c|c|c|c|c|c|} \hline
\multicolumn{1}{|V{2cm}|}{} & \multicolumn{2}{V{3cm}|}{Ordinary Partitioned Transition Relation} & \multicolumn{2}{V{3cm}|}{Transition Relation with TLEBDDs} & \multicolumn{2}{V{3cm}|}{Transition Relation with TLEBDDs and shift to target-state immediately} \\ \hline
\multicolumn{1}{|V{2cm}|}{Problem} & \multicolumn{1}{V{1.8cm}|}{Number of BDD Nodes} & \multicolumn{1}{V{1.2cm}|}{Time} & \multicolumn{1}{V{1.8cm}|}{Number of BDD Nodes} & \multicolumn{1}{V{1.2cm}|}{Time} & \multicolumn{1}{V{1.8cm}|}{Number of BDD Nodes} & \multicolumn{1}{V{1.2cm}|}{Time}\\ \hline
\multicolumn{1}{|V{2cm}|}{MutexLocal 5} & \multicolumn{1}{r|}{252,176} & \multicolumn{1}{r|}{34s} & \multicolumn{1}{r|}{176,577} & \multicolumn{1}{r|}{29s} & \multicolumn{1}{r|}{140,808} & \multicolumn{1}{r|}{3s}\\ \hline
\multicolumn{1}{|V{2cm}|}{MutexLocal 7} & \multicolumn{1}{r|}{6,618,487} & \multicolumn{1}{r|}{47:59m} & \multicolumn{1}{r|}{4,977,342} & \multicolumn{1}{r|}{44:05m} & \multicolumn{1}{r|}{4,090,041} & \multicolumn{1}{r|}{5:37m}\\ \hline
\multicolumn{1}{|V{2cm}|}{MutexLocal 8} & \multicolumn{1}{r|}{41,448,929} & \multicolumn{1}{r|}{7:45h}  & \multicolumn{1}{r|}{31,092,345} & \multicolumn{1}{r|}{7:10h} & \multicolumn{1}{r|}{25,704,013} & \multicolumn{1}{r|}{51:37m}\\ \hline
\multicolumn{1}{|V{2cm}|}{Peterson 5} & \multicolumn{1}{r|}{1,470,096} & \multicolumn{1}{r|}{6:32m} & \multicolumn{1}{r|}{720,661} & \multicolumn{1}{r|}{5:28m} & \multicolumn{1}{r|}{577,274} & \multicolumn{1}{r|}{49s}\\ \hline
\multicolumn{1}{|V{2cm}|}{Peterson 6} & \multicolumn{1}{r|}{11,051,785} & \multicolumn{1}{r|}{3:20h}  & \multicolumn{1}{r|}{8,562,251} & \multicolumn{1}{r|}{3:20h} & \multicolumn{1}{r|}{6,344,196} & \multicolumn{1}{r|}{24:40m}\\ \hline
\multicolumn{1}{|V{2cm}|}{Peterson 7} & \multicolumn{1}{r|}{$>$100,000,000} & \multicolumn{1}{r|}{$>$24h} & \multicolumn{1}{r|}{$>$100,000,000} & \multicolumn{1}{r|}{$>$24h} & \multicolumn{1}{r|}{89,401,785} & \multicolumn{1}{r|}{10:46h}\\ \hline
\multicolumn{1}{|V{2cm}|}{CCP 5} & \multicolumn{1}{r|}{205,449} & \multicolumn{1}{r|}{1:02m} & \multicolumn{1}{r|}{172,964} & \multicolumn{1}{r|}{57s} & \multicolumn{1}{r|}{117,875} & \multicolumn{1}{r|}{4s}\\ \hline
\multicolumn{1}{|V{2cm}|}{CCP 8} & \multicolumn{1}{r|}{9,840,064} & \multicolumn{1}{r|}{4:49h} & \multicolumn{1}{r|}{9,118,465} & \multicolumn{1}{r|}{4:35h} & \multicolumn{1}{r|}{5,855,155} & \multicolumn{1}{r|}{16:23m}\\ \hline
\multicolumn{1}{|V{2cm}|}{CCP 10} & \multicolumn{1}{r|}{$>$75,000,000} & \multicolumn{1}{r|}{$>$24h} & \multicolumn{1}{r|}{$>$75,000,000} & \multicolumn{1}{r|}{$>$24h} & \multicolumn{1}{r|}{67,822,819} & \multicolumn{1}{r|}{12:32h}\\ \hline
\multicolumn{1}{|V{2cm}|}{DP 15} & \multicolumn{1}{r|}{309,329} & \multicolumn{1}{r|}{3:55m} & \multicolumn{1}{r|}{294,403} & \multicolumn{1}{r|}{3:50m} & \multicolumn{1}{r|}{193,402} & \multicolumn{1}{r|}{16s}\\ \hline
\multicolumn{1}{|V{2cm}|}{DP 20} & \multicolumn{1}{r|}{3,614,204} & \multicolumn{1}{r|}{2:27h} & \multicolumn{1}{r|}{3,539,148} & \multicolumn{1}{r|}{2:27h} & \multicolumn{1}{r|}{2,267,828} & \multicolumn{1}{r|}{8:41m}\\ \hline
\multicolumn{1}{|V{2cm}|}{DP 22} & \multicolumn{1}{r|}{9,595,403} & \multicolumn{1}{r|}{9:06h} & \multicolumn{1}{r|}{9,446,319} & \multicolumn{1}{r|}{9:11h} & \multicolumn{1}{r|}{6,018,632} & \multicolumn{1}{r|}{32:34m}\\ \hline\end{tabular}
\label{table:ForwardReachability}
\end{table}

The first benchmark in Table \ref{table:ForwardReachability} is an extended simple Mutual Exclusion Algorithm. There, a critical section exists which can be reached by a component if a shared variable points to it. This benchmark has also other shared variables. They store for every control state of the components the number of components currently being in this control state. Additionally, every component has one local variable which stores the number of components currently being in the new control state when a component moves its control state. Our experimental results show that  big memory improvements can be achieved by using our TLEBDD to store the transition relations and we also see slight runtime improvements. The runtime improvements occur with TLEBDDs because we don't have to walk through edges of identity patterns in the recursion by using our new algorithm ANDRecursive. If we additionally shift the state variables immediately to the corresponding from-state variables, we even get further memory reductions and also large runtime improvements. The second testcase in Table \ref{table:ForwardReachability} is the Peterson Mutual Exclusion Protocol \cite{Peterson81}. It is a protocol where entry to the critical section is gained by a single process via a series of $n-1$ competitions. There is at least one loser for each competition and the protocol satisfies the mutual exclusion condition, since at most one process can win the final competition. Table \ref{table:ForwardReachability} shows that we achieve significant memory gains by just using  TLEBDDs as transition relation. If we additionally shift the state variables immediately, we can further reduce the peak number of live nodes and we get very large runtime improvements. Table \ref{table:ForwardReachability} also shows experimental results for the CCP Cache Coherence Protocol. It refers to a cache coherence protocol developed from S. German (see for example \cite{PnueliRuahZuck01AutomDeducVer}). As our experimental results show, we can slightly reduce the memory requirements by using a TLEBDD. When we shift the state variables immediately,  we get significant additional memory and runtime improvements. The last testcase in Table \ref{table:ForwardReachability} is the Dining Philosophers Problem (mentioned DP in Table \ref{table:ForwardReachability}). Our implementation is an imitation of the monitors solution from \cite{PrincipConcurProgBenAri06}. The experimental results show that the memory requirements can not be reduced very much by using TLEBDDs. Also a little runtime increase can be observed for $22$ components. This runtime increase can presumably be eliminated by optimizing the cache utilization. Nevertheless, significant memory and runtime savings can be observed again when we shift the state variables immediately.

\begin{table}[htbp]
\caption{Experimental results for building only the transition relation}
\centering
\begin{tabular}{|c|c|c|c|} \hline
\multicolumn{1}{|V{3cm}|}{} & \multicolumn{1}{V{3cm}|}{Ordinary Partitioned Transition Relation} & \multicolumn{1}{V{3cm}|}{Transition Relation only identity reduced} & \multicolumn{1}{V{2.4cm}|}{Transition Relation with TLEBDDs} \\ \hline
\multicolumn{1}{|V{3cm}|}{Problem} & \multicolumn{1}{V{3cm}|}{Number of BDD Nodes} &  \multicolumn{1}{V{3cm}|}{Number of BDD Nodes} &  \multicolumn{1}{V{2.4cm}|}{Number of BDD Nodes} \\ \hline
\multicolumn{1}{|V{3cm}|}{MutexLocal 75} & \multicolumn{1}{r|}{115,735,537} & \multicolumn{1}{r|}{10,105,558} &  \multicolumn{1}{r|}{141,623} \\ \hline
\multicolumn{1}{|V{3cm}|}{MutexLocal 255} & \multicolumn{1}{r|}{mem ov} & \multicolumn{1}{r|}{77,576,543} &  \multicolumn{1}{r|}{320,755} \\ \hline
\multicolumn{1}{|V{3cm}|}{MutexLocal 2047} & \multicolumn{1}{r|}{mem ov} & \multicolumn{1}{r|}{mem ov} &  \multicolumn{1}{r|}{3,414,118} \\ \hline
\multicolumn{1}{|V{3cm}|}{Peterson 8} & \multicolumn{1}{r|}{110,560,066} & \multicolumn{1}{r|}{47,415,495} &  \multicolumn{1}{r|}{17,403,225} \\ \hline
\multicolumn{1}{|V{3cm}|}{Peterson 9} & \multicolumn{1}{r|}{mem ov} & \multicolumn{1}{r|}{115,675,330} &  \multicolumn{1}{r|}{40,089,105} \\ \hline
\multicolumn{1}{|V{3cm}|}{Peterson 10} & \multicolumn{1}{r|}{mem ov} & \multicolumn{1}{r|}{mem ov} &  \multicolumn{1}{r|}{90,597,275} \\ \hline
\multicolumn{1}{|V{3cm}|}{CCP 18} & \multicolumn{1}{r|}{74,758,155} & \multicolumn{1}{r|}{74,728,268} &  \multicolumn{1}{r|}{9,840,393} \\ \hline
\multicolumn{1}{|V{3cm}|}{CCP 19} & \multicolumn{1}{r|}{mem ov} & \multicolumn{1}{r|}{mem ov} &  \multicolumn{1}{r|}{19,671,729} \\ \hline
\multicolumn{1}{|V{3cm}|}{CCP 21} & \multicolumn{1}{r|}{mem ov} & \multicolumn{1}{r|}{mem ov} &  \multicolumn{1}{r|}{78,656,120} \\ \hline
\end{tabular}
\label{table:TransRelMax}
\end{table}
Table \ref{table:TransRelMax} shows experimental results about the maximum number of components for which the transition relation can be built alone with different transition relation types. We there present experimental results for a standard partitioned transition relation, a partitioned transition relation which is only identity reduced and for a transition relation with TLEBDDs. As our experimental results show, the number of components for which the transition relation can be built can always be enlarged by using TLEBDDs. If we use only identity reduction, we can not increase the number of components as large as with TLEBDDs and we even don't get an increase in the number of components for the CCP testcase. This shows the efficiency of our TLEBDD approach. We omitted the experimental results for the dining philosophers testcase here, because it only has a small transition relation that can already be build with an ordinary partitioned transition relation for more than $1000$ components.

\section{Conclusion and Outlook}
In this paper we presented a new approach to store the transition relation of asynchronous systems. Our approach exploits the common property of BDD packages to store isomorphic subgraphs only once. The presented experimental results confirm that our approach can lead to big memory savings. This allows the verification of larger systems. Additionally, our method can enlarge the parts of the transition relation which can be stored in a single partition of a partitioned transition relation. In this way fewer nodes may be needed and verification can possibly be accelerated. Additionally, we presented a new algorithm to combine BDDs and TLEBDDs efficiently. As our experimental results confirm, an immediate shift to the from-state variables leads to very large runtime and memory reductions for interleaved variable orderings by using this new algorithm.

In the future we intend to investigate the usage of TLEBDDs for storing the transition relation with other state-space exploration algorithms than the traditional breadth-first algorithm. By using other algorithms, like e.g. breadth-first generation with chaining, or the saturation algorithm, possibly even greater memory savings may occur. To investigate the performance of the use of TLEBDDs with other verification benchmarks and state-space exploration algorithms we intend to implement their usage for the symbolic model checker NuSMV \cite{CimattiClarkeCAV02NuSMV2}. Also, we will try to investigate the consequences of different TLEBDD variable orderings on the memory requirements and the verification runtime.

\bibliographystyle{eptcs}
\bibliography{literatur}

\begin{thebibliography}{10}
\providecommand{\bibitemdeclare}[2]{}
\providecommand{\urlprefix}{Available at }
\providecommand{\url}[1]{\texttt{#1}}
\providecommand{\href}[2]{\texttt{#2}}
\providecommand{\urlalt}[2]{\href{#1}{#2}}
\providecommand{\doi}[1]{doi:\urlalt{http://dx.doi.org/#1}{#1}}
\providecommand{\bibinfo}[2]{#2}

\bibitemdeclare{book}{PrincipConcurProgBenAri06}
\bibitem{PrincipConcurProgBenAri06}
\bibinfo{author}{M.~Ben-Ari} (\bibinfo{year}{2006}):
  \emph{\bibinfo{title}{Principles of Concurrent and Distributed Programming
  (2nd Edition) (Prentice-Hall International Series in Computer Science)}}.
\newblock \bibinfo{publisher}{Addison-Wesley Longman Publishing Co., Inc.},
  \bibinfo{address}{Boston, MA, USA}.

\bibitemdeclare{inproceedings}{BranchingTimeLogic}
\bibitem{BranchingTimeLogic}
\bibinfo{author}{M.~Ben-Ari}, \bibinfo{author}{Z.~Manna} \&
  \bibinfo{author}{A.~Pnueli} (\bibinfo{year}{1981}): \emph{\bibinfo{title}{The
  temporal logic of branching time}}.
\newblock In: {\sl \bibinfo{booktitle}{POPL '81: Proceedings of the 8th ACM
  SIGPLAN-SIGACT symposium on Principles of programming languages}},
  \bibinfo{publisher}{ACM}, pp. \bibinfo{pages}{164--176},
  \doi{10.1145/567532.567551}.

\bibitemdeclare{article}{GraphBasedAlgorithmsBryant86}
\bibitem{GraphBasedAlgorithmsBryant86}
\bibinfo{author}{R.~E. Bryant} (\bibinfo{year}{1986}):
  \emph{\bibinfo{title}{Graph-Based Algorithms for Boolean Function
  Manipulation}}.
\newblock {\sl \bibinfo{journal}{IEEE Transactions on Computers}}
  \bibinfo{volume}{35}, pp. \bibinfo{pages}{677--691},
  \doi{10.1109/TC.1986.1676819}.

\bibitemdeclare{inproceedings}{SymbModCheckWiPartTransRel91}
\bibitem{SymbModCheckWiPartTransRel91}
\bibinfo{author}{J.~R. Burch}, \bibinfo{author}{E.~M. Clarke} \&
  \bibinfo{author}{D.~E. Long} (\bibinfo{year}{1991}):
  \emph{\bibinfo{title}{Symbolic Model Checking with Partitioned Transition
  Relations}}.
\newblock \bibinfo{publisher}{North-Holland}, pp. \bibinfo{pages}{49--58}.

\bibitemdeclare{inproceedings}{SatuBasedReachConDis05}
\bibitem{SatuBasedReachConDis05}
\bibinfo{author}{G.~Ciardo} \& \bibinfo{author}{A.~J. Yu}
  (\bibinfo{year}{2005}): \emph{\bibinfo{title}{Saturation-based symbolic
  reachability analysis using conjunctive and disjunctive partitioning}}.
\newblock In: {\sl \bibinfo{booktitle}{Proc. CHARME, LNCS 3725}},
  \bibinfo{publisher}{Springer-Verlag}, pp. \bibinfo{pages}{146--161}.

\bibitemdeclare{inproceedings}{CimattiClarkeCAV02NuSMV2}
\bibitem{CimattiClarkeCAV02NuSMV2}
\bibinfo{author}{A.~Cimatti}, \bibinfo{author}{E.~Clarke},
  \bibinfo{author}{E.~Giunchiglia}, \bibinfo{author}{F.~Giunchiglia},
  \bibinfo{author}{M.~Pistore}, \bibinfo{author}{M.~Roveri},
  \bibinfo{author}{R.~Sebastiani} \& \bibinfo{author}{A.~Tacchella}
  (\bibinfo{year}{2002}): \emph{\bibinfo{title}{{NuSMV Version 2: An OpenSource
  Tool for Symbolic Model Checking}}}.
\newblock In: {\sl \bibinfo{booktitle}{Proc. International Conference on
  Computer-Aided Verification (CAV 2002)}}, {\sl \bibinfo{series}{LNCS}}
  \bibinfo{volume}{2404}, \bibinfo{publisher}{Springer},
  \bibinfo{address}{Copenhagen, Denmark}.

\bibitemdeclare{inproceedings}{TempModCheckEmCl}
\bibitem{TempModCheckEmCl}
\bibinfo{author}{E.~M. Clarke} \& \bibinfo{author}{E.~A. Emerson}
  (\bibinfo{year}{1982}): \emph{\bibinfo{title}{Design and Synthesis of
  Synchronization Skeletons Using Branching-Time Temporal Logic}}.
\newblock In: {\sl \bibinfo{booktitle}{Logic of Programs, Workshop}},
  \bibinfo{publisher}{Springer-Verlag}, \bibinfo{address}{London, UK}, pp.
  \bibinfo{pages}{52--71}, \doi{10.1007/BFb0025774}.

\bibitemdeclare{book}{ClarkePeledGrumModCheck00}
\bibitem{ClarkePeledGrumModCheck00}
\bibinfo{author}{E.~M. Clarke}, \bibinfo{author}{O.~Grumberg} \&
  \bibinfo{author}{D.~A. Peled} (\bibinfo{year}{2000}):
  \emph{\bibinfo{title}{Model checking}}.
\newblock \bibinfo{publisher}{MIT Press}.

\bibitemdeclare{inproceedings}{Bryant03SymbRepresOrdFuncTempl}
\bibitem{Bryant03SymbRepresOrdFuncTempl}
\bibinfo{author}{A.~Goel}, \bibinfo{author}{G.~Hasteer} \&
  \bibinfo{author}{R.~Bryant} (\bibinfo{year}{2003}):
  \emph{\bibinfo{title}{Symbolic representation with ordered function
  templates}}.
\newblock In: {\sl \bibinfo{booktitle}{Proceedings of the 40th annual Design
  Automation Conference}}, \bibinfo{series}{DAC '03}, \bibinfo{publisher}{ACM},
  \bibinfo{address}{New York, NY, USA}, pp. \bibinfo{pages}{431--435},
  \doi{10.1145/775832.775946}.

\bibitemdeclare{inproceedings}{KaiserKroeningWahl10DynCuttoff}
\bibitem{KaiserKroeningWahl10DynCuttoff}
\bibinfo{author}{Alexander Kaiser}, \bibinfo{author}{Daniel Kroening} \&
  \bibinfo{author}{Thomas Wahl} (\bibinfo{year}{2010}):
  \emph{\bibinfo{title}{Dynamic Cutoff Detection in Parameterized Concurrent
  Programs}}.
\newblock In: {\sl \bibinfo{booktitle}{Computer-Aided Verification (CAV)}},
  \doi{10.1007/978-3-642-14295-6_55}.

\bibitemdeclare{article}{MultiValuedDecDiag98}
\bibitem{MultiValuedDecDiag98}
\bibinfo{author}{T.~Kam}, \bibinfo{author}{T.~Villa},
  \bibinfo{author}{R.~Brayton} \& \bibinfo{author}{A.~Sangiovanni-Vincentelli}
  (\bibinfo{year}{1998}): \emph{\bibinfo{title}{Multi-valued decision diagrams:
  theory and applications}}.
\newblock {\sl \bibinfo{journal}{Multiple-Valued Logic}}
  \bibinfo{volume}{4(1-2)}, pp. \bibinfo{pages}{9--62}.

\bibitemdeclare{book}{SymbModCheckMcM93}
\bibitem{SymbModCheckMcM93}
\bibinfo{author}{K.~L. McMillan} (\bibinfo{year}{1993}):
  \emph{\bibinfo{title}{Symbolic Model Checking}}.
\newblock \bibinfo{publisher}{Kluwer Academic Publishers},
  \bibinfo{address}{Norwell, MA, USA}.

\bibitemdeclare{inproceedings}{EffSollMatrixDiagMill01}
\bibitem{EffSollMatrixDiagMill01}
\bibinfo{author}{A.~S. Miner} (\bibinfo{year}{2001}):
  \emph{\bibinfo{title}{Efficient solution of GSPNs using Canonical Matrix
  Diagrams}}.
\newblock In: {\sl \bibinfo{booktitle}{Proceedings of the 9th International
  Workshop on Petri Nets and Performance Models}}, \bibinfo{publisher}{IEEE
  Comp. Soc. Press}, pp. \bibinfo{pages}{101--110},
  \doi{10.1109/PNPM.2001.953360}.

\bibitemdeclare{inproceedings}{SatGenClassMod04}
\bibitem{SatGenClassMod04}
\bibinfo{author}{A.~S. Miner} (\bibinfo{year}{2004}):
  \emph{\bibinfo{title}{Saturation for a General Class of Models}}.
\newblock In: {\sl \bibinfo{booktitle}{QEST '04: Proceedings of the The
  Quantitative Evaluation of Systems, First International Conference}},
  \bibinfo{publisher}{IEEE Computer Society}, \bibinfo{address}{Washington, DC,
  USA}, pp. \bibinfo{pages}{282--291}, \doi{10.1109/QEST.2004.38}.

\bibitemdeclare{article}{Peterson81}
\bibitem{Peterson81}
\bibinfo{author}{G.~L. Peterson} (\bibinfo{year}{1981}):
  \emph{\bibinfo{title}{Myths About the Mutual Exclusion Problem.}}
\newblock {\sl \bibinfo{journal}{Inf. Process. Lett.}}
  \bibinfo{volume}{12}(\bibinfo{number}{3}), pp. \bibinfo{pages}{115--116},
  \doi{10.1016/0020-0190(81)90106-X}.

\bibitemdeclare{article}{LTLPnueli}
\bibitem{LTLPnueli}
\bibinfo{author}{A.~Pnueli} (\bibinfo{year}{1981}): \emph{\bibinfo{title}{A
  temporal logic of concurrent programs}}.
\newblock {\sl \bibinfo{journal}{Theoretical Computer Science}}
  \bibinfo{volume}{13}, pp. \bibinfo{pages}{45--60}.

\bibitemdeclare{inproceedings}{PnueliRuahZuck01AutomDeducVer}
\bibitem{PnueliRuahZuck01AutomDeducVer}
\bibinfo{author}{A.~Pnueli}, \bibinfo{author}{S.~Ruah} \&
  \bibinfo{author}{L.~Zuck} (\bibinfo{year}{2001}):
  \emph{\bibinfo{title}{Automatic Deductive Verification with Invisible
  Invariants}}.
\newblock \bibinfo{publisher}{Springer}, pp. \bibinfo{pages}{82--97}.

\bibitemdeclare{inproceedings}{TempModCheckQuSif}
\bibitem{TempModCheckQuSif}
\bibinfo{author}{J.-P. Queille} \& \bibinfo{author}{J.~Sifakis}
  (\bibinfo{year}{1982}): \emph{\bibinfo{title}{Specification and verification
  of concurrent systems in CESAR}}.
\newblock In: {\sl \bibinfo{booktitle}{DesProceedings of the 5th Colloquium on
  International Symposium on Programming}},
  \bibinfo{publisher}{Springer-Verlag}, \bibinfo{address}{London, UK}, pp.
  \bibinfo{pages}{337--351}, \doi{10.1007/3-540-11494-7_22}.

\bibitemdeclare{inproceedings}{Ranjan95EfficientBDDalgorithms}
\bibitem{Ranjan95EfficientBDDalgorithms}
\bibinfo{author}{R.~K. Ranjan}, \bibinfo{author}{A.~Aziz},
  \bibinfo{author}{R.~K. Brayton}, \bibinfo{author}{B.~Plessier} \&
  \bibinfo{author}{C.~Pixley} (\bibinfo{year}{1995}):
  \emph{\bibinfo{title}{Efficient BDD algorithms for FSM synthesis and
  verification}}.
\newblock In: {\sl \bibinfo{booktitle}{In IEEE/ACM Proceedings International
  Workshop on Logic Synthesis, Lake Tahoe (NV}}.

\bibitemdeclare{manual}{Somenzi09Cudd}
\bibitem{Somenzi09Cudd}
\bibinfo{author}{F.~Somenzi} (\bibinfo{year}{2009}):
  \emph{\bibinfo{title}{CUDD: CU Decision Diagram Package, Release 2.4.2}}.
\newblock \bibinfo{organization}{University of Colorado at Boulder},
  \bibinfo{address}{http://vlsi.colorado.edu/~fabio/CUDD/}.

\bibitemdeclare{inproceedings}{WahlBlancEmerson08Sviss}
\bibitem{WahlBlancEmerson08Sviss}
\bibinfo{author}{T.~Wahl}, \bibinfo{author}{N.~Blanc} \&
  \bibinfo{author}{A.~Emerson} (\bibinfo{year}{2008}):
  \emph{\bibinfo{title}{Sviss: Symbolic Verification of Symmetric Systems}}.
\newblock In: {\sl \bibinfo{booktitle}{Tools and Algorithms for the
  Construction and Analysis of Systems (TACAS)}}.

\bibitemdeclare{inproceedings}{YangBryantBiere98CavBDDPerformanceStudy}
\bibitem{YangBryantBiere98CavBDDPerformanceStudy}
\bibinfo{author}{B.~Yang}, \bibinfo{author}{R.~E. Bryant},
  \bibinfo{author}{D.~R. O'Hallaron}, \bibinfo{author}{A.~Biere},
  \bibinfo{author}{O.~Coudert}, \bibinfo{author}{G.~Janssen},
  \bibinfo{author}{R.~K. Ranjan} \& \bibinfo{author}{F.~Somenzi}
  (\bibinfo{year}{1998}): \emph{\bibinfo{title}{A Performance Study of
  BDD-Based Model Checking}}.
\newblock In: {\sl \bibinfo{booktitle}{Proceedings of the Second International
  Conference on Formal Methods in Computer-Aided Design}},
  \bibinfo{series}{FMCAD '98}, \bibinfo{publisher}{Springer-Verlag},
  \bibinfo{address}{London, UK}, pp. \bibinfo{pages}{255--289}.

\end{thebibliography}
\end{document}